%% file: main.tex
\newcommand{\gd}[1]{}
\newcommand{\arxiv}[1]{#1}

\documentclass[a4paper,USenglish,cleveref,autoref,thm-restate]{lipics-v2021}

\arxiv{\pdfoutput=1}
\arxiv{\hideLIPIcs}

\nolinenumbers

\usepackage{todonotes}
\usepackage{mathtools}
\usepackage{tabularray}

\usepackage{apptools}
\gd{\newcommand{\restateref}[1]{$\star$}}
\arxiv{\newcommand{\restateref}[1]{\IfAppendix{\hyperref[#1]{$\star$}}{\hyperref[#1*]{$\star$}}}}
\newenvironment{proofsketch}{\renewcommand{\proofname}{Proof (sketch)}\proof}{\endproof}

\DeclarePairedDelimiter{\ceils}{\lceil}{\rceil}

\newcommand{\paramstyle}[1]{\mathrm{#1}}
\newcommand{\crs}{\paramstyle{cr}}
\newcommand{\lcr}{\paramstyle{lcr}}

\newcommand{\nd}{\paramstyle{nd}}
\newcommand{\vi}{\paramstyle{vi}}
\newcommand{\vc}{\paramstyle{vc}}
\newcommand{\td}{\paramstyle{td}}

\newcommand{\fvs}{\paramstyle{fvs}}
\newcommand{\fes}{\paramstyle{fes}}
\newcommand{\mw}{\paramstyle{mw}}
\newcommand{\dn}{\paramstyle{dn}}

\newcommand{\Bandwidth}{\textsc{Bandwidth}}
\newcommand{\TwoSidedkPlanarity}{\textsc{Two-Sided $k$-Planarity}}
\newcommand{\BipartiteCRProb}{\textsc{Bipartite Crossing Number}}
\newcommand{\CRProb}{\textsc{Crossing Number}}
\newcommand{\OPTest}{\textsc{$1$-Planarity Testing}}
\newcommand{\kPTest}{\textsc{$k$-Planarity Testing}}
\newcommand{\UnaryBinPacking}{\textsc{Unary Bin Packing}}

\title{Structural Parameterizations of $k$-Planarity} %

\author{Tatsuya Gima}{Hokkaido University, Japan}{gima@ist.hokudai.ac.jp}{https://orcid.org/0000-0003-2815-5699}{JSPS KAKENHI Grant Number JP24K23847 and JP25K03077.}

\author{Yasuaki Kobayashi}{Hokkaido University, Japan}{koba@ist.hokudai.ac.jp}{https://orcid.org/0000-0003-3244-6915}{JSPS KAKENHI Grant Numbers JP23K28034, JP24H00686, and JP24H00697.}

\author{Yuto Okada}{Nagoya University, Japan}{pv.20h.3324@s.thers.ac.jp}{https://orcid.org/0000-0002-1156-0383}{Supported by JST SPRING, Grant Number JPMJSP2125.}

\authorrunning{T. Gima, Y. Kobayashi, and Y. Okada} %

\Copyright{Tatsuya Gima, Yasuaki Kobayashi, and Yuto Okada} %



\gd{\relatedversiondetails{Full Version}{https://arxiv.org/abs/2506.10717}}

\acknowledgements{We would like to thank Alexander Wolff for valuable discussions in W\"urzburg and anonymous reviewers for insightful comments. }%

\EventEditors{John Q. Open and Joan R. Access}
\EventNoEds{2}
\EventLongTitle{42nd Conference on Very Important Topics (CVIT 2016)}
\EventShortTitle{CVIT 2016}
\EventAcronym{CVIT}
\EventYear{2016}
\EventDate{December 24--27, 2016}
\EventLocation{Little Whinging, United Kingdom}
\EventLogo{}
\SeriesVolume{42}
\ArticleNo{23}

\begin{document}

\maketitle

\begin{abstract}
The concept of $k$-planarity is extensively studied in the context of Beyond Planarity. A graph is $k$-planar if it admits a drawing in the plane in which each edge is crossed at most $k$ times. The local crossing number of a graph is the minimum integer $k$ such that it is $k$-planar. The problem of determining whether an input graph is $1$-planar is known to be NP-complete even for near-planar graphs [Cabello and Mohar, SIAM J. Comput. 2013], that is, the graphs obtained from planar graphs by adding a single edge. Moreover, the local crossing number is hard to approximate within a factor $2 - \varepsilon$ for any $\varepsilon > 0$ [Urschel and Wellens, IPL 2021]. To address this computational intractability, Bannister, Cabello, and Eppstein [JGAA 2018] investigated the parameterized complexity of the case of $k = 1$, particularly focusing on structural parameterizations on input graphs, such as treedepth, vertex cover number, and feedback edge number. In this paper, we extend their approach by considering the general case $k \ge 1$ and give (tight) parameterized upper and lower bound results. In particular, we strengthen the aforementioned lower bound results to subclasses of constant-treewidth graphs: we show that testing $1$-planarity is NP-complete even for near-planar graphs with \emph{feedback vertex set number at most~3} and \emph{pathwidth at most~4}, and the local crossing number is hard to approximate within \emph{any constant factor} for graphs with \emph{feedback vertex set number at most~2}.
\end{abstract}

\section{Introduction}\label{sec:intro}

Graph drawing is recognized as an important area in the research of graph theory and algorithms due to various real-world applications.
In particular, efficient algorithms for drawing graphs on the plane without edge crossings have received considerable attention from various perspectives.
However, to visualize real-world networks, it is necessary to consider drawing \emph{non-planar} graphs in many cases.
This research direction is positioned as Beyond Planarity~\cite{DidimoLM19,Dujmovic00PF24,HT2020}, and a significant amount of effort has been dedicated to analyzing their graph-theoretic properties and designing algorithms for drawing these graphs.
In this context, the problem of drawing an input graph on the plane with the minimum number of crossings is one of the most well-studied problems.
In other words, the problem asks for the \emph{crossing number} of an input graph, which is known to be NP-hard~\cite{GareyJ83}.
Among various studies on this problem, the (recent) progress on the parameterized complexity of \CRProb, where the goal is to determine whether an input graph $G$ can be drawn in the plane with at most $k$ crossings, is remarkable~\cite{VerdiereM21,Grohe04,KawarabayashiR07,LokshtanovP0S0Z25}.
In particular, this problem is \emph{fixed-parameter tractable} (FPT) when parameterized by $k$, that is, it admits an algorithm with running time $f(k)n^{O(1)}$, where $n$ is the number of vertices in the input graph and $f$ is a computable function.

The \emph{local crossing number} is one of the well-studied variations of the standard crossing number~\cite{Ackerman19,BannisterCE18,CabelloM13,GrigorievB07,HoffmannLRT20,KorzhikM13,PachT97,UrschelW21}, as witnessed by the fact that it is selected as the topic for the live challenge\footnote{\url{https://mozart.diei.unipg.it/gdcontest/2025/live/}} held in conjunction with GD~2025.
Let $k$ be an integer.
We say that a graph $G$ is \emph{k-planar} if it can be drawn in the plane so that each edge involves at most $k$ crossings.
The minimum integer $k$ such that $G$ admits a $k$-planar drawing is called the \emph{local crossing number} of $G$, denoted by $\lcr(G)$.
The problem of deciding if $\lcr(G) \le k$ for a given graph $G$ is called \kPTest.
Unlike \CRProb, it is already NP-complete to decide whether $\lcr(G) \le 1$~\cite{BannisterCE18,CabelloM13,GrigorievB07,KorzhikM13}.
More strongly, this problem is NP-complete even when the input is restricted to planar graphs with an extra edge~\cite{CabelloM13} or graphs with constant bandwidth~\cite{BannisterCE18}.
Urschel and Wellens~\cite{UrschelW21} later extended the NP-hardness by showing that \kPTest{} is NP-complete for any fixed $k \geq 1$, even if the input is restricted to graphs with local crossing number at most $k$ or at least $2k$.
This implies that, unless P $=$ NP, there is no polynomial-time $(2-\varepsilon)$-approximation for local crossing number for any $\varepsilon > 0$.

To overcome such an intractability, Bannister, Cabello, and Eppstein~\cite{BannisterCE18} pursued the parameterized complexity of the case of $k = 1$, namely \OPTest{}, by focusing on \emph{structural parameterizations}.
They showed that \OPTest{} is FPT parameterized by treedepth and by feedback edge number (see~\cref{sec:preli} for definitions).
As mentioned above, \OPTest{} is NP-complete even on the classes of graphs with bounded bandwidth.
This intractability is indeed inherited by wider classes of graphs, such as those with bounded cliquewidth, treewidth, and pathwidth.
In this direction, M{\"{u}}nch, Pfister, and Rutter~\cite{MunchPR24} recently considered \kPTest{} on special cases of pathwidth-bounded graphs and gave exact and approximation algorithms.

\begin{table}[t]
\centering
\caption{The table summarizes our and known results.
The columns ``unbounded'', ``parameter'', and ``$k = 1$'' indicate the results when $k$ is taken as input, as parameter, and $k = 1$, respectively.}\label{tbl:results}
\begin{tblr}{
    cells = {halign = c, valign=m},
    hline{1} = {2-Z}{},
    hline{2-Z},
    vline{1} = {2-Z}{},
    vline{2-Z} = {1-Z}{}, 
    cell{1}{1,2,3} = {r=1, c=1}{},
    cell{2}{2} = {r=1, c=2}{bg=green!18}, %
    cell{2}{4} = {bg=green!18}, %
    cell{3}{2} = {r=1, c=3}{bg=red!28}, %
    cell{4}{2} = {bg=red!10}, %
    cell{4}{3} = {bg=green!18}, %
    cell{4}{4} = {bg=green!18}, %
    cell{5}{2} = {r=2,c=1}{bg=red!10}, %
    cell{5}{3} = {r=2,c=2}{bg=green!18}, %
    cell{7}{2} = {r=1,c=3}{bg=red!10}, %
    cell{8}{2} = {r=1,c=3}{bg=red!28}, %
}
 & $k$: unbounded & $k$: parameter & $k = 1$ \\ 
feedback edge set number  & FPT (Thm. \ref{thm:fpt:fes}) & &  FPT \cite{BannisterCE18} \\ 
feedback vertex set number & NP-complete when $\fvs = 2$ (Thm. \ref{thm:1-planarity-np-complete-fvs-2})\\ 
treedepth            & {W[1]-hard \\ (Thm. \ref{thm:k-planarity-w-hard-td})} & {FPT (Thm. \ref{thm:fpt:td-k}) \\ (non-uniform~FPT~\cite{Zehavi22})} & FPT \cite{BannisterCE18} \\ 
longest induced path & {W[1]-hard \\ (Thm. \ref{thm:hardness:tcn})} & FPT (Cor. \ref{cor:exPnFPT}) \\ 
twin cover number & W[1]-hard  & \\ 
distance to path forest & W[1]-hard (Thm. \ref{thm:1-planarity-w1-hard-distance-to-path-forest})\\ 
domination number & NP-complete when $\dn = 2$ (Cor. \ref{cor:1-planarity-npc-domination-number-2})
\end{tblr}
\end{table}

\begin{figure}
  \begin{subfigure}[b]{.49\textwidth}
    \centering
    \includegraphics[page=1]{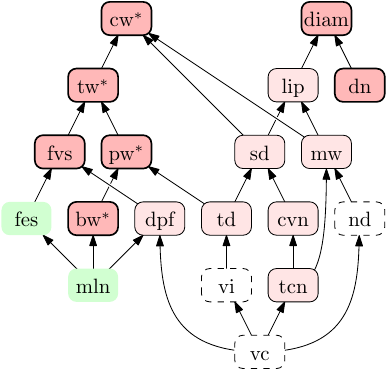}
    \subcaption{}
    \label{fig:hierarchy:without-k}
  \end{subfigure}
  \hfill
  \begin{subfigure}[b]{.49\textwidth}
    \centering
    \includegraphics[page=2]{figure/hierarchy.pdf}
    \subcaption{}
    \label{fig:hierarchy:plus-k}
  \end{subfigure}
  \caption{A visualization of \cref{tbl:results} of the cases where $k$ is (a) unbounded and (b) parameter. No borderline, dotted borderline, normal borderline, bold borderline mean that the complexity parameterized by the parameter is FPT, unknown, W[1]-hard, paraNP-complete, respectively. The complexities for parameters with ($*$) are already known results and the others are our results. An arrow indicates that the upper parameter is bounded from above by a function of the lower parameter. See the footnote\footref{footnote:parameter-names} for the abbreviations of parameter names.}
  \label{fig:hierarchies}
\end{figure}

In this paper, we extend results of \cite{BannisterCE18} by considering the general case $k \ge 1$ and show a fine parameterized complexity landscape with respect to graph width parameters (see \cref{tbl:results} and \cref{fig:hierarchies} for a summary of our results\footnote{In \cref{fig:hierarchies}, the names are abbreviated as follows: cw is clique width, diam is diameter, tw is treewidth, lip is longest induced path, dn is domination number, fvs is feedback vertex set, pw is pathwidth, sd is shrub-depth, mw is modular-width, fes is feedback edge set, bw is bandwidth, dpf is distance to path forest, td is treedepth, cvn is cluster vertex deletion number, nd is neighborhood diversity, mln is max leaf number, vi is vertex integrity, tcn is twin cover number, and vc is vertex cover number.\label{footnote:parameter-names}}).
We extend algorithms of \cite{BannisterCE18} by showing that \kPTest{} is FPT parameterized by feedback edge set number (\cref{thm:fpt:fes}) and by treedepth plus $k$ (\cref{thm:fpt:td-k}).
More generally, we show that \kPTest{} is FPT on $P_t$-free graphs when parameterized by $t + k$ (\cref{cor:exPnFPT}).
We also give polynomial kernelizations for \kPTest{} with respect to vertex cover number (\cref{thm:vc-kernel}) and neighborhood diversity (\cref{thm:nd-kernel}).
On the negative side, we show that \kPTest{} is W[1]-hard when parameterized solely by treedepth (\cref{thm:k-planarity-w-hard-td}), which indicates that \cref{thm:fpt:td-k} is tight in a certain sense, and by twin cover number (\cref{thm:hardness:tcn}).
We also show that \OPTest{} is W[1]-hard parameterized by distance to path forest (\cref{thm:1-planarity-w1-hard-distance-to-path-forest}).
Similar to \cite{CabelloM13}, the last result also proves that \OPTest{} remains NP-complete even on planar graphs with a single additional edge.
Our reduction shows that \OPTest{} is NP-complete even when further adding restrictions that the input graph has feedback vertex set number at most~3.
Using a completely different reduction, we show that \OPTest{} is NP-complete even if the input is restricted to graphs with feedback vertex set number~2 (\cref{thm:1-planarity-np-complete-fvs-2}). 
This reduction also shows that there is no constant factor approximation for \kPTest{} unless P $=$ NP, which significantly strengthens the $(2 - \varepsilon)$-inapproximability result of Urschel and Wellens~\cite{UrschelW21}.

\gd{Several proofs (marked with $\star$) are omitted and available in the full version.}
\arxiv{Several proofs (marked with clickable $\star$) are deferred to the Appendix.}

\section{Preliminaries}\label{sec:preli}
Let $G$ be a graph.
The vertex set and edge set of $G$ are denoted by $V(G)$ and $E(G)$, respectively.
For a vertex set $X \subseteq V(G)$ (resp.\ edge set $X \subseteq E(G)$), $G - X$ denotes the graph obtained from $G$ by deleting all elements in $X$.
The subgraph of $G$ induced by $X \subseteq V(G)$ is denoted by $G[X]$.
For $v \in V(G)$, $N_G(v)$ denotes the set of neighbors of $v$.

A (\emph{topological}) \emph{drawing} $\Gamma$ of $G$ is a representation in the plane that maps the vertices of $G$ to distinct points in the plane and each edge of $G$ to a non-self-intersecting (Jordan) curve connecting the points corresponding to the endpoints.
In the rest of this paper, we may simply refer to these points as vertices and these curves as edges when no confusion is possible. 
A \emph{crossing} in $\Gamma$ is an intersection of distinct edges that is not a common endpoint.
We assume that, in any drawing, an edge does not pass through any vertex other than its endpoints, and no three (or more) edges cross at a common point.
The \emph{crossing number} of $G$, denoted by $\crs(G)$, is the minimum number of crossings over all possible drawings of $G$.
For an integer $k$, a drawing $\Gamma$ is said to be \emph{$k$-planar} if each edge is crossed at most $k$ times.
A graph is said to be $k$-planar if it admits a $k$-planar drawing.
The \emph{local crossing number} of $G$, denoted by $\lcr(G)$, is the minimum integer $k$ such that $G$ is $k$-planar.
Throughout this paper, we use the following property of local crossing number.

\begin{restatable}[\restateref{lem:subdivision}]{lemma}{lemSubdivision}
\label{lem:subdivision}
Let $G$ be a graph and $k$ be a positive integer.
Let $G_k$ be the graph obtained from $G$ by subdividing each edge $k - 1$ times.
Then, $\lcr(G) \le k$ if and only if $\lcr(G_k) \le 1$.
\end{restatable}

In this paper, we present parameterized algorithms and hardness results with respect to several graph width parameters.
However, some of them may not need to be defined precisely.
Hence, we only provide the definitions that are directly needed to present our results.
For basic concepts in parameterized complexity, we refer the reader to \cite{Aoihon}.

Let $G$ be a graph.
A \emph{vertex cover} of $G$ is a set of vertices $S$ such that $G - S$ is edge-less, and the \emph{vertex cover number} of $G$ is the minimum size of a vertex cover of $G$, which is denoted by $\vc(G)$.
The \emph{feedback vertex set number} (resp.\ \emph{feedback edge set number}) of $G$, denoted by $\fvs(G)$ (resp.\ $\fes(G)$), is the minimum cardinality of a vertex set (resp.\ edge set) $X$ such that $G - X$ is acyclic.
An \emph{elimination forest} of $G$ is a rooted forest $F$ (i.e., a set of rooted trees) such that $V(F) = V(G)$ and for every edge in $G$, one of the endpoints is an ancestor of the other in $F$.
The \emph{treedepth} of $G$, denoted by $\td(G)$, is the minimum integer $k$ such that $G$ has an elimination forest of height~$k$.
Two vertices $u$ and $v$ are called \emph{true twins} if they are adjacent and $N_G(u) = N_G(v)$; they are called \emph{false twins} if they are non-adjacent and $N_G(u) = N_G(v)$.
The \emph{neighborhood diversity} of $G$ is at most $k$ if the vertex set of $G$ can be partitioned into $k$ sets $V_1, \dots, V_k$ such that either all the pairs in $V_i$ are true twins or they are false twins.
Each set $V_i$ is called a \emph{twin class}.
Note that each true twin class induces a clique and each false twin class induces an independent set in $G$.
The neighborhood diversity of $G$ is denoted by $\nd(G)$.
A vertex set $S \subseteq V(G)$ is called a \emph{twin cover} of $G$ if each connected component of $G - S$ consists of true twins in $G$.
The \emph{twin cover number} of $G$ is the minimum size of a twin cover of $G$.
The \emph{max leaf number} of $G$ is the maximum number of leaves of a spanning tree of $G$.

\input{hardness}

\input{fpt_algo}

\input{poly_kernel}

\section{Conclusion}
In this paper, we study the parameterized complexity of \kPTest{} from the perspective of graph structural parameterizations and show several (tight) upper and lower bound results.
We leave several interesting open problems relevant to our results. 

When $k$ is considered as input, we have only shown that \kPTest{} is FPT parameterized by feedback edge set number. 
It would be interesting to seek similar results using other graph parameters. 
In this direction, a highly related problem \CRProb{} is known to be FPT parameterized by vertex cover number~\cite{HlinenyS19}.
However, it seems that their approach cannot be directly applied to \kPTest{}, requiring new insights for this problem.
Towards this, showing intractability with a general parameter of vertex cover number, such as vertex integrity or neighborhood diversity (see \cref{fig:hierarchy:without-k}), would also be a nice open problem.

\bibliography{ref}

\arxiv{
    \newpage
    \appendix
    \include{appendix}
}

\end{document}

%% file: hardness.tex
\section{Hardness}

In this section, we give several intractability results for \kPTest{} and \OPTest{}.
In particular, we show that \kPTest{} is hard to approximate within any constant factor in polynomial time, even on graphs with feedback vertex set number at most~2, and \OPTest{} is NP-complete even on near-planar graphs with feedback vertex set number~3, where a graph is \emph{near-planar} if it can be obtained from a planar graph by adding a single edge.
These two results improve the previous results of \cite{UrschelW21} and \cite{CabelloM13}.
To this end, we design two completely different reductions, one is given from \TwoSidedkPlanarity{} and the other is given from \UnaryBinPacking{}, which also yield several consequences other than these two results.

\subsection{Inapproximability for Graphs with Feedback Vertex Set Number~2}

In this subsection, we show that it is NP-hard to approximate the local crossing number of a given graph within any constant factor, even if the graph has feedback vertex set number 2.

Before describing our reductions, we start with a technical lemma, showing that by adding some vertices and edges to a graph, one can impose a certain form on its $k$-planar drawings.

Let $G$ be a graph with $m$ edges and let $k$ be a positive integer.
Let $X \subseteq V(G)$ be a nonempty vertex subset.
We define a new graph $G'$ by adding a vertex $r$ and $km + 1$ paths of length~2 between $r$ and $v$ for each $v \in X$ to $G$.\footnote{The length of a path is defined as the number of edges in it.}
We refer to these paths of length~2 as \emph{spokes}.

\begin{restatable}[\restateref{lem:crossing-free-spokes}]{lemma}{LemCrossingFreeSpokes}
\label{lem:crossing-free-spokes}
    Suppose that $G'$ has a $k$-planar drawing $\Gamma$.
    Then, there is a $k$-planar drawing $\Gamma'$ of $G'$ such that each spoke has no crossings and the subdrawings of $\Gamma$ and $\Gamma'$ induced by $G$ are identical.
\end{restatable}

\begin{proofsketch}
    Since the edges in $G$ can cross other edges $km$ times in total, for each $u \in X$, there is a spoke between $r$ and $u$ that does not cross any edges of $G$.
    Let $P_u$ be such a spoke for each $u \in X$ and $\mathcal{P} \coloneqq \{P_u \mid u \in X\}$.
    Although spokes in $\mathcal{P}$ may cross, we can untangle them, preserving the subdrawing for $G$.
    Then we redraw the other spokes along $\mathcal{P}$.
\end{proofsketch}

We will use this lemma in the following form.

\begin{corollary}\label{cor:crossing-free-spokes}
    Let $S \subseteq V(G)$.
    Let $G'$ be a graph obtained from $G$ by adding a vertex $s$ and $k|E(G)| + 1$ spokes between $s$ and each $u \in S$.
    Moreover, let $S' \subseteq V(G')$ and let $G''$ be a graph obtained from $G'$ by adding a vertex $s'$ and $k|E(G')| + 1$ spokes between $s'$ and each $u \in S'$.
    Suppose that $G''$ has a $k$-planar drawing.
    Then, there is a $k$-planar drawing of $G''$ such that all spokes have no crossings.
\end{corollary}
\begin{proof}
    By~\cref{lem:crossing-free-spokes}, $G''$ has a $k$-planar drawing $\Gamma'$ in which all the spokes incident to $s'$ have no crossings.
    Let $\Gamma$ be the subdrawing of $\Gamma'$ induced by $G'$.
    As shown in the proof of \cref{lem:crossing-free-spokes}, for each $u \in S$, there is a spoke $P_u$ between $s$ and $u$ that has no crossing with any edge of $G$ in $\Gamma$.
    Since this spoke $P_u$ does not cross any spoke incident to $s'$, the uncrossing operation and redrawing spokes between $s$ and $u$ used in \cref{lem:crossing-free-spokes} never create a new crossing with spokes incident to $s'$.
    Thus, by applying \cref{lem:crossing-free-spokes} to $G'$ and $\Gamma$, we have a $k$-planar drawing of $G''$ being claimed.
\end{proof}

We now turn to our reduction from \TwoSidedkPlanarity{}, in which we are given a bipartite graph $G = (X \cup Y, E)$ with two independent sets $X$ and $Y$, and an integer $k$.
The goal is to determine whether $G$ has a \emph{2-layer $k$-planar drawing}, which is a special case of a $k$-planar drawing in which the vertices in $X$ are drawn on a line, the vertices in $Y$ are drawn on another line parallel to the first line, and each edge is drawn as a straight line.

The idea of our reduction is the same as that of Garey and Johnson~\cite{GareyJ83} from \BipartiteCRProb{} to \CRProb{}.
We remark that, however, the proof is more involved due to the difficulty of $k$-planarity.

Let $\langle G = (X, Y, E), k \rangle$ be an instance of \TwoSidedkPlanarity{}, where $n_X = |X|$, $n_Y = |Y|$, and $m = |E|$.
We construct a graph $G'$ from $G$ as follows (see \cref{fig:reduction:fvs:image}):
\begin{enumerate}
    \item we introduce two vertices $u_X, u_Y$;
    \item for each $x \in X$, add $\ell_1 \coloneqq km + 1$ spokes between $u_X$ and $x$;
    \item for each $y \in Y \cup \{u_X\}$, add $\ell_2 \coloneqq k(\ell_1  n_X + m) + 1$ spokes between $u_Y$ and $y$.
\end{enumerate}

\begin{figure}
    \centering
    \includegraphics[page=3]{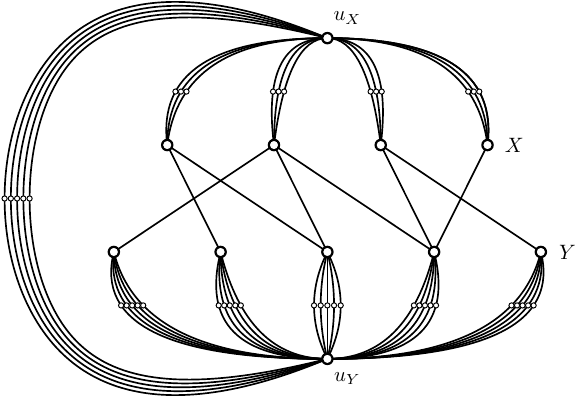}
    \caption{An illustration of the graph $G'$ constructed from an instance $\langle G, k \rangle$ of \TwoSidedkPlanarity{}. Note that the numbers of spokes are different from that of the actual construction.}
    \label{fig:reduction:fvs:image}
\end{figure}

Then, $\langle G', k\rangle$ is the instance we construct for \kPTest{}.
It is easy to verify that the above construction can be done in polynomial time.

\begin{lemma}\label{lem:reduction}
    The graph $G$ admits a 2-layer $k$-planar drawing if and only if the graph $G'$ admits a $k$-planar drawing.
\end{lemma}

\begin{proof}
    It is clear that we can obtain a $k$-planar drawing of $G'$ from a 2-layer $k$-planar drawing of $G$ by drawing $u_X$, $u_Y$, and spokes as \cref{fig:reduction:fvs:image}.
    Hence, we only show the other direction.

    Suppose that the graph $G'$ admits a $k$-planar drawing $\Gamma'$.
    By~\cref{cor:crossing-free-spokes}, we can assume that all the spokes in $G'$ have no crossings in $\Gamma'$.
    We then replace the set of spokes connecting two vertices with a single edge between them and contract the edge between $u_X$ and $u_Y$ by identifying them into a new vertex $u$ in $\Gamma'$.
    We let $G^*$ denote the graph obtained in this way.
    As all spokes have no crossings, the obtained drawing of $G^*$, denoted by $\Gamma^*$, is still $k$-planar.
    Let us note that $G^*$ is isomorphic to the graph obtained from $G$ by adding a universal vertex $u$ (which is a vertex adjacent to all the vertices in $G$).
    We claim that we can draw a closed curve $C$ passing through all the vertices in $V(G^*) \setminus \{u\}$ such that
    \begin{enumerate}
        \item no edge of $G^*$ intersects $C$ (i.e., $C$ is a noose in $\Gamma^*$);
        \item all the edges incident to $u$ are contained in one of the two regions $R$ with boundary $C$;
        \item all the edges in $E(G)$ are contained in the other region with boundary $C$;
        \item all the vertices in $X$ appear consecutively on $C$;
        \item all the vertices in $Y$ appear consecutively on $C$.
    \end{enumerate}
    To this end, we draw a curve $C$ while keeping track of the cyclic ordering of the edges incident to $u$ (and hence the neighbors of $u$) in $\Gamma^*$.
    As these incident edges have no crossings, we can draw such a curve satisfying (1), (2), and (3).
    To see the properties (4) and (5), suppose for contradiction that there are $x, x' \in X$ and $y, y' \in Y$ such that $x, y, x', y'$ appear in this order on $C$.
    Due to the property~(2), $C$ can be seen as a closed curve on $\Gamma'$ such that the region corresponding to $R$ contains all spokes of $G'$.
    Let $P_x$ be a path in $G'$ between $x$ and $x'$ via $u_X$ consisting only of spokes.
    We define $P_y$ similarly.
    In $\Gamma'$ those paths are contained in $R$ and $C$.
    However, since the vertices $x, x', y, y'$ are all placed on $C$, such a drawing must yield a crossing, which contradicts the fact that spokes in $G'$ have no crossings in $\Gamma'$.

    Let $X = \{x_1, \dots, x_{n_X}\}$ and let $Y = \{y_1, \dots, y_{n_Y}\}$ such that $x_1, \dots, x_{n_X}, y_1, \dots, y_{n_Y}$ appear on $C$ in this order.
    We now convert the drawing $\Gamma^*$ into a drawing $\hat{\Gamma}$ such that 
    \begin{itemize}
        \item all the vertices in $V(G^*) \setminus \{u\}$ are placed on a line $\ell$ in the order $x_1, \dots, x_{n_X}, y_1, \dots, y_{n_Y}$;
        \item all the edges incident to $u$ are drawn as straight segments in a half plane separated by $\ell$;
        \item all the other edges are drawn in the other half plane in such a way that the curve corresponding to an edge forms a semicircular arc between the endpoints.
    \end{itemize}
    See~\cref{fig:reduction:fvs:G-and-u:arc-diagram} for an illustration.
    \begin{figure}
      \begin{subfigure}[b]{.49\textwidth}
        \centering
        \includegraphics[page=4]{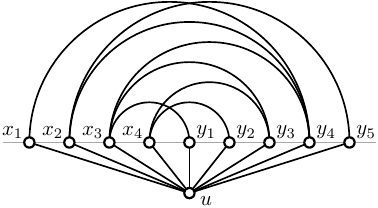}
        \subcaption{}
        \label{fig:reduction:fvs:G-and-u:arc-diagram}
      \end{subfigure}
      \hfill
      \begin{subfigure}[b]{.49\textwidth}
        \centering
        \includegraphics[page=5]{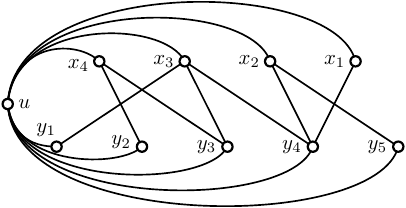}
        \subcaption{}
        \label{fig:reduction:fvs:G-and-u:2-layer-drawing}
      \end{subfigure}
      \caption{Two topologically equivalent drawings of $G^*$, where the subdrawing induced by $G$ is (a)~an arc diagram, and (b)~a 2-layer drawing.}
      \label{fig:reduction:fvs:G-and-u}
    \end{figure}
    Observe that this drawing $\hat{\Gamma}$ is also $k$-planar as two edges of $G^*$ cross in $\hat{\Gamma}$ only if they cross in $\Gamma^*$ as well.
    Moreover, $\hat{\Gamma}$ is \emph{simple}, meaning that no pair of edges crosses more than once and no two edges incident to a common vertex cross.
    This in turn can be converted into a $k$-planar drawing $\Gamma$ as in \cref{fig:reduction:fvs:G-and-u:2-layer-drawing}. The subdrawing of $\Gamma$ induced by $G^* - \{u\}$ is a $2$-layer $k$-planar drawing such that two edges in $G^* - \{u\}$ cross if and only if they cross in $\hat{\Gamma}$.
    Hence, $G$ ($=G^* - \{u\}$) admits a 2-layer $k$-planar drawing.
\end{proof}

This reduction immediately gives us the following consequence.
Note that \cref{thm:1-planarity-np-complete-fvs-2} is optimal in the sense that every graph with feedback vertex set number $1$ is planar.

\begin{theorem}\label{thm:1-planarity-np-complete-fvs-2}
    \OPTest{} is NP-complete even if the given graph has feedback vertex set number~$2$.
\end{theorem}

\begin{proof}
    The NP-membership of \OPTest{} is shown in \cite{KorzhikM13}.
    It is known that \TwoSidedkPlanarity{} is NP-complete even on trees~\cite[Theorem 11]{Recognizing2-LayerandOuterk-PlanarGraphs}.
    Thus, the graph $G'$ constructed in the proof of \cref{lem:reduction} has feedback vertex set number at most~2.
    This in turn implies the claim by \cref{lem:subdivision}.
\end{proof}

The above reduction leads to further consequences.
To see this, we first sketch the reduction used in \cite[Theorem 11]{Recognizing2-LayerandOuterk-PlanarGraphs}.
They performed a polynomial-time reduction from \Bandwidth{} to \TwoSidedkPlanarity{}.
Let $T$ be a tree, $b$ be a positive integer, and $\ell$ be an arbitrary even number with $\ell \ge 2b^2$.
For each edge in $T$, each edge $e \in E(T)$ is subdivided once by introducing a new vertex $w_e$, and we add $\ell$ leaves adjacent to each original vertex in $T$.
They showed that there is an ordering $v_1, \dots, v_{|V(T)|}$ of $V(T)$ such that $|i - j| \le b$ for any edge $\{v_i, v_j\} \in E(T)$ if and only if the obtained graph $T'$ has a 2-layer $k$-planar drawing with $k = (\ell + 4)(b-1)/2$.

\begin{theorem}\label{thm:k-planarity-w-hard-td}
    \kPTest{} is W[1]-hard with respect to treedepth, even on graphs with feedback vertex set number 2.
\end{theorem}

\begin{proof}
    It is known that the problem \Bandwidth{} is W[1]-hard with respect to treedepth, even on trees~\cite{DBLP:journals/tcs/GimaHKKO22}.
    The above reduction we sketched~\cite[Theorem 11]{Recognizing2-LayerandOuterk-PlanarGraphs} also preserves the boundedness of treedepth.
    To see this, we convert an elimination forest $F^*$ of $T$ to that of $T'$ as follows.
    Starting from $F^*$, we add the $\ell$ leaves attached to each vertex as leaf children of it.
    For each edge $e \in E(T)$, one of the end vertices, say $v$, is a descendant of the other in $F^*$.
    We then add $w_e$ as a leaf child of $v$ in $F^*$.
    It is easy to verify that the rooted forest obtained in this way is an elimination forest of $T'$ and its height increases by at most~1 compared to the original elimination forest of $T$, yielding that $\td(T') \le \td(T) + 1$.
    Hence, \TwoSidedkPlanarity{} is also W[1]-hard with respect to treedepth, even on trees.
    Moreover, the reduction used in \cref{lem:reduction} increases the treedepth by at most $3$, which implies the claim.
\end{proof}

\begin{lemma}\label{lem:two-sided-inapproximability}
    Unless P $=$ NP, for any constant $c \geq 1$, there is no polynomial-time $c$-approximation algorithm for \TwoSidedkPlanarity{} on trees.
\end{lemma}

\begin{proof}
    For any constant $c \ge 1$, it is NP-hard to distinguish the following cases: the bandwidth of $T$ is at most $b$ or more than $cb$ for a given tree $T$ and an integer $b$~\cite{HardnessApproximatingBandwidth}.
    We set $\ell = 2t - 4$ for sufficiently large $t$ and construct a tree $T'$ according to the above reduction from $T$.
    If the bandwidth of $T$ is at most $b$, then we have $\lcr(T') \le t(b - 1)$.
    Otherwise, we have $\lcr(T') \ge tcb$.
    This implies that a polynomial-time $c$-approximation algorithm for \TwoSidedkPlanarity{} on trees distinguishes the above two cases.
\end{proof}

Combining \cref{lem:reduction} and \cref{lem:two-sided-inapproximability}, we obtain the following.

\begin{theorem}\label{thm:kptest:inapprox}
    Unless P $=$ NP, for any constant $c \geq 1$, there is no polynomial-time $c$-approximation algorithm for \kPTest{}, even if the given graph has feedback vertex set number~2.
\end{theorem}

\subsection{Distance to Path Forest}\label{subsec:distance-to-path-forest}

In this subsection, we show that \OPTest{} is W[1]-hard parameterized by distance to path forest, i.e., the minimum number of vertices required to make the graph into a disjoint union of paths.
The underlying idea of the proof is similar to the one in \cite{GrigorievB07}.
We perform a reduction from \UnaryBinPacking{}.
An instance $I$ of \UnaryBinPacking{} consists of a (multi)set of positive integers $S$ and positive integers $B$ and $b$ such that $\sum_{x \in S} x = bB$ and $B = O(\mathrm{poly}(|S|))$.
Then the goal is to determine if it is possible to partition $S$ into $b$ sets so that the sum of each set is exactly $B$.
\UnaryBinPacking{} is NP-complete when $b$ is given as input~\cite{GareyJ83} and W[1]-hard when parameterized by $b$~\cite{JansenKMS13}.

\begin{theorem}\label{thm:1-planarity-w1-hard-distance-to-path-forest}
    \OPTest{} is W[1]-hard with respect to distance to path forest.
\end{theorem}
\begin{proof}
    Let $\langle S = \{x_1, \dots, x_{|S|}\}, B, b \rangle$ be an instance of \UnaryBinPacking{}.
    In the following, we assume that $b \ge 3$.
    We construct an instance $\langle G = (V, E) \rangle$ of \OPTest{} as follows.
    Let $G$ be a graph consisting of vertices $u_1, u_2, v_1, \dots, v_{b}$, and paths $P_1, \dots, P_{|S|}$ with $|V(P_i)| = x_i$ for $1 \le i \le |S|$.
    Let $\mathcal P = \{P_1, \dots, P_{|S|}\}$.
    We add a path of length~$B$ between $v_i$ and $v_{i+1}$ for each $1 \leq i \leq b$, where $v_{b+1} \coloneqq v_1$.
    These paths form a cycle $C$ of $bB$ vertices, passing through $v_i$ for all~$i$.
    Next, we add an edge between $u_i$ and each vertex in $P_j$ for $1 \le i \le 2$ and $1 \le j \le |S|$.
    Let $m$ be the number of edges in the graph constructed so far (namely, $m = 4bB - |S|$).
    Finally, for $1 \le i \le b$, we add $\ell_1 \coloneqq m + 1$ spokes between $u_1$ and $v_i$ and add $\ell_2 \coloneqq (\ell_1 b + m) + 1$ spokes between $u_2$ and $v_i$.
    See \cref{fig:reduction:distance-to-path-forest:image} for an illustration of the constructed graph $G$.
    Observe that, by construction, removing $b+2$ vertices, say $u_1, u_2, v_1, \dots, v_b$, yields a path forest.
    The construction can be done in polynomial time.
    \begin{figure}
        \centering
        \includegraphics[page=1]{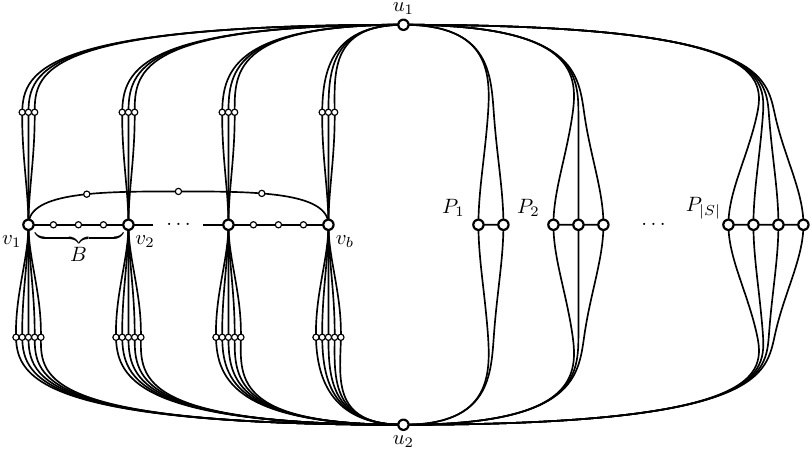}
        \caption{An illustration of the reduction for \cref{thm:1-planarity-w1-hard-distance-to-path-forest}. Note that the number of depicted spokes between $u_i$ and $v_j$ is not accurate for aesthetic purposes.}
        \label{fig:reduction:distance-to-path-forest:image}
    \end{figure}

    From a solution of \UnaryBinPacking{}, we can obtain a $1$-planar drawing of $G$ easily, as in \cref{fig:reduction:distance-to-path-forest:1-planar-drawing}.
    \begin{figure}
        \centering
        \includegraphics[page=2]{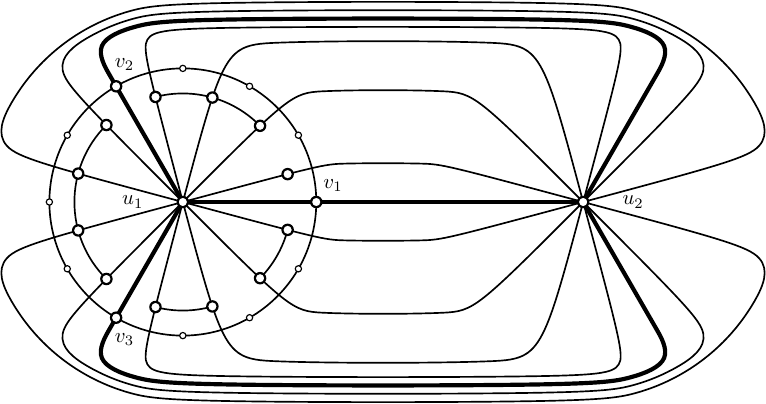}
        \caption{A 1-planar drawing of the graph constructed from an instance $\langle S = \{1, 2, 2, 3, 4\}, B = 4, b = 3 \rangle$ of \UnaryBinPacking{}. Each thick edge represents the spokes between its endpoints.}
        \label{fig:reduction:distance-to-path-forest:1-planar-drawing}
    \end{figure}
    To be more precise, let $\{S_1, \dots, S_b\}$ be a solution for \UnaryBinPacking{}.
    We first draw the cycle $C$ as a circle in such a way that the spokes between $u_1$ and $v_i$ partition the interior of the circle into $b$ regions $R_1, \dots, R_b$, where $R_i$ is enclosed by a spoke between $u_1$ and $v_i$, a spoke between $u_1$ and $v_{i + 1}$, and a path between $v_i$ and $v_{i + 1}$ on $C$.
    For each region $R_i$, we draw a path $P_j \in \mathcal P$ inside $R_i$ if $x_j \in S_i$.
    Since the sum of the numbers of vertices of the paths drawn in $R_j$ is exactly $B$, we can draw those paths so that each edge in the path between $v_i$ and $v_{i+1}$ is crossed exactly once.
    Hence, we can draw $G$ so that each edge is crossed at most once.

    For the other direction, suppose that $G$ has a 1-planar drawing.
    By \cref{cor:crossing-free-spokes}, there exists a 1-planar drawing $\Gamma$ of the graph $G$ such that all the spokes in $G$ have no crossings.
    Since the spokes have no crossings in $\Gamma$, the subdrawing induced by them is planar, and hence there are two vertices $v_i$ and $v_j$ that belong to the outer cycle of this subdrawing.
    As $b \ge 3$, all the other $v_{\ell}$'s are drawn inside the region bounded by this outer cycle.
    Moreover, since two consecutive $v_\ell$ and $v_{\ell + 1}$ are connected by a path in $C$, the two distinguished vertices $v_i$ and $v_j$ must be consecutive, namely $j = (i + 1) \bmod b$.
    By appropriately renaming those vertices, we can assume that $G - \bigcup_{1 \le i \le |S|} V(P_i)$ are drawn as in \cref{fig:reduction:distance-to-path-forest:skelton}.
    Then, there are $b$ internally vertex-disjoint paths $Q_1, \dots, Q_b$ between $u_1$ and $u_2$ such that $Q_i$ is a concatenation of spokes between $u_1$ and $v_i$ and between $v_i$ and $u_2$ for each $i$.
    These paths separate the plane into $b$ regions, $R_1, \dots, R_b$ as \cref{fig:reduction:distance-to-path-forest:regions}.
    \begin{figure}
      \begin{subfigure}[b]{.49\textwidth}
        \centering
        \includegraphics[page=6]{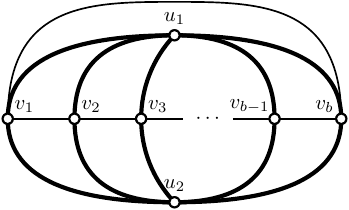}
        \subcaption{}
        \label{fig:reduction:distance-to-path-forest:skelton}
      \end{subfigure}
      \hfill
      \begin{subfigure}[b]{.49\textwidth}
        \centering
        \includegraphics[page=7]{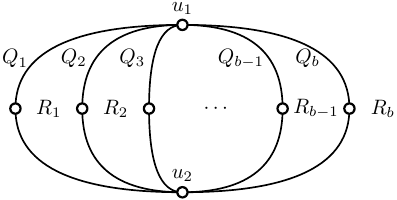}
        \subcaption{}
        \label{fig:reduction:distance-to-path-forest:regions}
      \end{subfigure}
      \caption{The subgraphs induced by (a)~cycle~$C$ and spokes (boldlines) and (b)~the paths $Q_1, \dots, Q_b$, which defines $b$ regions $R_1, \dots, R_b$.}
    \end{figure}
    As each path $P_i$ cannot lie in two different regions, $\mathcal P$ can be partitioned into $\{\mathcal{P}_1, \dots, \mathcal{P}_b\}$, where $\mathcal{P}_i \subseteq \mathcal P$ is the collection of paths drawn inside $R_i$.
    Observe that, for each path $P_j \in \mathcal{P}_i$, there are $x_i = |V(P_i)|$ edge-disjoint paths between $u_1$ and $u_2$ that are drawn inside $R_i$.
    Each of them must cross the path between $v_i$ and $v_{i+1}$ on $C$, implying that there are at most $B$ such paths inside $R_i$ as $\Gamma$ is 1-planar.
    Since there are exactly $b$ regions, each $\mathcal P_i$ must satisfy $\sum_{P \in \mathcal P_i} |V(P)| = B$.
    Therefore, we can partition $S$ into $b$ sets such that the sum of each set is exactly $B$.
\end{proof}

Observe that the above reduction still works even if we add an edge between $u_1$ and each vertex in $C$, in which case the vertex set $\{u_1, u_2\}$ forms a dominating set.

\begin{corollary}\label{cor:1-planarity-npc-domination-number-2}
    \OPTest{} is NP-complete even if the domination number is 2.
\end{corollary}

As mentioned in \cref{sec:intro}, Cabello and Mohar~\cite{CabelloM13} showed that \OPTest{} is NP-hard even on near-planar graphs.
\cref{thm:1-planarity-w1-hard-distance-to-path-forest} strengthens their hardness result as the graph $G$ constructed in the reduction is a near-planar graph with the following properties.

\begin{corollary}\label{cor:hardness:near-planar}
    \OPTest{} is NP-complete on near-planar graphs with feedback vertex set number at most~3 and pathwidth at most~4.
\end{corollary}

\subsection{Twin Cover Number}

In this subsection, we show that \kPTest{} is W[1]-hard with respect to twin cover number, which complements the positive result of the ``$+k$ setting'' to be presented in the subsequent section.

We give a similar reduction from \UnaryBinPacking{} as in \cref{subsec:distance-to-path-forest}.
However, in order to substitute cliques for paths, we first show the W[1]-hardness of a restricted version of \UnaryBinPacking, where integers in $S$ are sufficiently small that the constructed graph admits a $B$-planar drawing.

\begin{restatable}[\restateref{lem:UBP:moderate-size-item}]{lemma}{LemUBPModerateSizeItem}
\label{lem:UBP:moderate-size-item}
\UnaryBinPacking{} is W[1]-hard parameterized by the number $b$ of bins, even if each integer $x \in S$ satisfies $x \leq \sqrt{B} - 1$.
\end{restatable}

\begin{proofsketch}
    We add $2bB$ copies of $B+1$ to $S$ of an instance of \UnaryBinPacking{} and increase the capacity of a bin by $2B(B+1)$ to obtain an instance with the claimed condition.
    We can show that there is no other solution than to distribute those $B+1$'s evenly to $b$ bins, making it equivalent to the original instance.
\end{proofsketch}

\begin{theorem}\label{thm:hardness:tcn}
    \kPTest{} is W[1]-hard parameterized by twin cover number.
\end{theorem}

\begin{proof}
    The reduction is almost analogous to the one used in \cref{thm:1-planarity-w1-hard-distance-to-path-forest}.
    Let $I = \langle S, B, b \rangle$ be an instance of \UnaryBinPacking{}.
    We can assume that $b \ge 3$ and $x \le \sqrt{B} - 1$ for all $x \in S$ due to \cref{lem:UBP:moderate-size-item}.
    While the basic construction of the graph $G$ is the same as in \cref{thm:1-planarity-w1-hard-distance-to-path-forest}, it differs in the following points:
    \begin{itemize}
        \item the cycle $C$ contains exactly $b$ vertices $v_1, \dots, v_b$;
        \item for each $x_i \in S$, $G$ contains a clique $C_i$ with $x_i$ vertices, instead of a path $P_i$, such that each vertex in the clique is adjacent to both $u_1$ and $u_2$;
        \item for each $v_i$, we add $\ell_1 \coloneqq Bm + 1$ spokes between $u_1$ and $v_i$ and add $\ell_2 \coloneqq B(b\ell_1 + m) + 1$ spokes between $u_2$ and $v_2$, where $m$ is the number of edges in $G$ that are not contained in the spokes.
    \end{itemize}
    When we draw $G$ as in \cref{fig:reduction:distance-to-path-forest:1-planar-drawing}, edges incident to a vertex in clique~$C_i$ may have crossings.
    Each edge $e$ in $C_i$ may cross other edges in $C_i$ and all the edges between $u_j$ and a vertex in $C_i$ for $1 \le j \le 2$.
    This implies that the number of crossings that $e$ involves is at most
    \begin{align*}
        |E(C_i)| - 1 + 2|V(C_i)| \le (\sqrt{B}-1)(\sqrt{B} - 2)/2 - 1 + 2(\sqrt{B} - 1) < B.
    \end{align*}
    Similarly, we can conclude that each edge between $u_j$ and a vertex in $C_i$ has at most $B - 1$ crossings.
    Since we can assume that all the spokes in $G$ have no crossings due to \cref{cor:crossing-free-spokes}, these spokes define exactly $b$ regions, each of which contains exactly $B$ vertices of cliques.
    Thus, we can show that $I$ is a yes-instance if and only if $G$ is $B$-planar as well.

    As each connected component in $G - (V(C) \cup \{u_1, u_2\})$ consists of true twins, $G$ has a twin cover of size at most $b + 2$.
    This completes the proof.
\end{proof}

%% file: fpt_algo.tex
\section{FPT Algorithms} \label{sec:fpt}
In contrast to the previous section, we mainly focus on the algorithmic aspects of \kPTest{} in this section and show that \kPTest{} is FPT when several well-known graph parameters (and $k$) are given as parameters.

\subsection{Reducing to \OPTest{}}
There are several consequences from \cref{lem:subdivision}, combining known FPT algorithms of \OPTest{} of \cite{BannisterCE18}.

Let $G$ be a graph and let $k$ be an integer.
Let $G_k$ be the graph obtained from $G$ by subdividing each edge $k - 1$ times.
Clearly, we have $\fes(G) = \fes(G_k)$.
By~\cref{lem:subdivision}, $G$ is $k$-planar if and only if $G_k$ is 1-planar.
Since \OPTest{} is FPT parameterized by feedback edge set number~\cite{BannisterCE18}, the following theorem holds.

\begin{theorem}\label{thm:fpt:fes}
    \kPTest{} is FPT parameterized by feedback edge set number.
\end{theorem}

We next observe that the treedepth of $G_k$ is not much larger than the treedepth of $G$.

\begin{restatable}[\restateref{lem:fpt:td-Gk}]{lemma}{LemFptTdGk}
\label{lem:fpt:td-Gk}
    For an integer $k \ge 0$, it holds that $\td(G_k) \le \td(G) + \ceils{\log_2 k}$.
\end{restatable}

Similarly to the above, we have the following theorem.
\begin{theorem}\label{thm:fpt:td-k}
    \kPTest{} is FPT parameterized by $\text{treedepth} + k$.
\end{theorem}

We would like to note that a similar result is already mentioned in \cite{Zehavi22}.
Since $k$-planarity is closed under vertex and edge deletions, it is known that the class of $k$-planar graphs with bounded treedepth is well-quasi-ordered by the induced-subgraph relation~\cite{NesetrilM15}. 
This implies a non-uniform version of \cref{thm:fpt:td-k}: \kPTest{} can be solved by just checking whether $G$ has an induced subgraph that belongs to a finite set $\mathcal F_{k, \td}$ of graphs, where the size of $\mathcal F_{k, \td}$ depends only on $k$ and $\td(G)$.

\Cref{thm:fpt:td-k} leads to FPT results for other parameters. To explain this, we introduce the notion of degeneracy.
For an integer $d \ge 0$, a graph $G$ is said to be \emph{$d$-degenerate} if each nonempty subgraph of $G$ contains a vertex of degree at most $d$.
A graph class $\mathcal C$ is \emph{degenerate} if there exists $d$ such that every graph in $\mathcal C$ is $d$-degenerate.

\begin{proposition}[{\cite[Proposition 6.4]{SparsityNM}}]\label{prop:td-exPn}
    A hereditary class of graphs $\mathcal C$ has bounded treedepth if and only if $\mathcal C$ is degenerate and does not contain a path $P_t$ for some $t$ as an induced subgraph.
\end{proposition}

It is known that each $k$-planar graph with $n$ vertices has at most $3.81\sqrt{k} n$ edges~\cite{Ackerman19}.
Since $k$-planarity is closed under taking a subgraph, the class of $k$-planar graphs is degenerate for every fixed~$k$.
Hence, combining with \cref{prop:td-exPn}, \kPTest{} on $P_t$-free graphs parameterized by $k + t$ can be reduced to the one parameterized by $\text{treedepth}+k$.
It is known that a graph class $\mathcal C$ does not contain a path $P_t$ for some $t$ as an induced subgraph if $\mathcal C$ has bounded shrub-depth~\cite{GanianHNOM19} or bounded modular-width~\cite{KimLMP18}.

\begin{corollary}\label{cor:exPnFPT}
    Let $\mathcal C$ be a hereditary class of graphs that excludes some path $P_t$ as an induced subgraph.
    Then, \kPTest{} over $\mathcal C$ is FPT parameterized by $k + t$.
    In particular, \kPTest{} is FPT parameterized by 
    $\text{shrub-depth} +k$,
    and $\text{modular-width} +k$.
\end{corollary}

%% file: poly_kernel.tex
\subsection{Polynomial Kernels}\label{subsec:poly-kernel}
In this subsection, we describe a kernelization algorithm for \kPTest{} parameterized by $\text{vertex cover number} + k$.
The high-level ideas follow the polynomial kernel for \OPTest{} due to Bannister, Cabello, and Eppstein~\cite{BannisterCE18}.

Let $G$ be a graph and let $S$ be a vertex cover of $G$.
We can find such a vertex cover $S$ of size at most twice the optimum in linear time.
First, we observe that the ``diversity'' of vertices outside of a vertex cover $S$ is sufficiently small when $G$ is $k$-planar.
For a function $f$, let $\#_{f}(X)$ denote $|\{f(x) \mid x \in X\}|$, the number of possible images of $X$ through $f$.
\begin{lemma}\label{lem:vc:deg2_class}
    Let $I_{\geq 2} \subseteq V(G) \setminus S$ be the set of vertices with degree at least~2 in $V(G) \setminus S$. 
    Let $\pi$ be a function that maps each vertex $v \in I_{\ge 2}$ to an unordered pair of its distinct neighbors (that is, $\pi\colon I_{\ge 2} \to \binom{N_G(v)}{2}$).
    If $\#_{\pi}(I_{\ge 2}) > |S|\cdot 3.81\sqrt{2k}$, then $G$ is not $k$-planar.
\end{lemma}
\begin{proof}
Let $G_S$ be a graph with $V(G_S) = S$ and $E(G_S) = \pi(I_{\ge 2})$.
Let $G'$ be the graph obtained from $G_S$ by subdividing each edge once.
Then $G'$ is a subgraph of $G$, and hence $G'$ is $k$-planar.
This implies that by~\cref{lem:subdivision}, $G_S$ is $2k$-planar.
Hence, $\#_{\pi}(I_{\ge 2}) = |E(G_S)| \leq |S| \cdot 3.81\sqrt{2k}$,
where the last inequality is shown in \cite{Ackerman19}.
\end{proof}

It is known that, for a positive integer $k$, if $G$ contains $K_{7k+1, 3}$ as a subgraph then $G$ is not $k$-planar~\cite{CzapH21,Pfister2025}.
Hence, the following also immediately holds.
\begin{corollary}\label{cor:vc:false_twins-deg3}
  Let $u \in V(G) \setminus S$ be a vertex of degree at least~3.
  If there are at least $7k+1$ false twins of $u$ (including $u$ itself), then the graph $G$ is not $k$-planar.
\end{corollary}

Next, we consider vertices of degree 2.
The following lemma is a key ingredient of our kernelization.
It is worth mentioning that, although a similar reduction rule is used in the case $k = 1$~\cite{BannisterCE18}, the proof for the general case $k \ge 1$ is considerably more involved. 

\begin{lemma}\label{lem:vc:false_twins-deg2}
  Let $u \in V(G) \setminus S$ be a vertex of degree~2 with $N_G(u) = \{s_1, s_2\} \subseteq S$.
  If there are more than $16k^2|S|$ vertices in $V(G) \setminus S$ with the same neighborhood $\{s_1, s_2\}$, then $G$ is $k$-planar if and only if $G-\{u\}$ is $k$-planar.
\end{lemma}
\begin{proof}
    The ($\Rightarrow$) direction is clear.
    Suppose that $G' \coloneqq G - \{u\}$ is $k$-planar and let $\Gamma$ be a $k$-planar drawing of $G'$.
    We can assume that $G'$ has no vertex of degree~1.
    
    Let $T_u = \{v \in V(G') \setminus S \mid N_{G'}(v) = \{s_1, s_2\}\}$.
    Then, we have $t_u = |T_u| \geq 16k^2 |S|$. %
    Let $\mathcal{P}$ be the set of paths between $s_1$ and $s_2$ via a vertex in $T_u$.
    Let us align the paths in $\mathcal{P}$ according to the cyclic order in $\Gamma$ around $s_1$, as $P_0, P_1, \dots, P_{t_u-1}$.
    For two paths $P_i, P_j$, the \emph{cyclic difference} of $P_i$ and $P_j$ is defined as $\min\{j - i, i - j + t_u\}$.
    We then claim the following.

    \begin{claim}\label{claim:cyclic-difference-geq-4k-not-crossing}
        Two paths $P_i, P_j \in \mathcal P$ with the cyclic difference at least $4k$ do not cross in $\Gamma$.
    \end{claim}

    \begin{claimproof}
        \begin{figure}
            \centering
            \includegraphics[page=1]{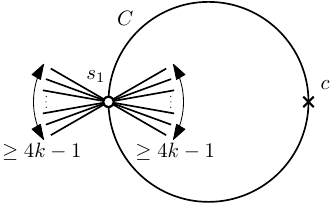}
            \caption{The Jordan curve $C$, the crossing point $c$, and the vertex $s_1$ with edges going out to each of the interior and exterior of $C$ due to the cyclic difference condition.}
            \label{fig:cyclic-difference-geq-4k}
        \end{figure}
        Suppose that $P_i$ and $P_j$ do cross in $\Gamma$.
        Let $c$ be a crossing point between $P_i$ and $P_j$, and let $C$ be a Jordan curve consisting of two arcs: the arc of $P_i$ from $s_1$ to $c$ and the arc of $P_j$ from $c$ to $s_1$.
        Such a crossing point $c$ can be chosen so that $C$ forms a closed curve.
        In both cases where $s_2$ is in the interior or in the exterior of regions bounded by $C$, from the condition that the cyclic difference of $P_i$ and $P_j$ is at least $4k$, there are at least $4k-1$ edges that must cross $C$ as shown in \cref{fig:cyclic-difference-geq-4k}.
        As $P_i, P_j$ can have at most $2k-1$ crossings other than $c$ each, and hence $4k-2$ crossings in total, this leads to a contradiction.
    \end{claimproof}

    \begin{figure}
        \centering
        \includegraphics{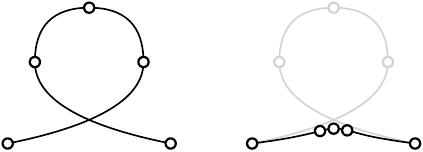}
        \caption{Let $P$ be a path whose internal vertices have degree exactly 2.
        A ``self-crossing'' on the path can be removed by suppressing the ``loop'' formed by the edges in $P$.}
        \label{fig:self-crossing}
    \end{figure}
    
    Let $\mathcal{P}' = \{P_{4ki} \mid 0 \leq i < 4k|S|\}$.
    As the cyclic difference of any pair of two paths in $\mathcal{P}'$ is at least $4k$, by \cref{claim:cyclic-difference-geq-4k-not-crossing}, no two paths in $\mathcal{P}'$ cross.
    We can also assume that each path in $\mathcal{P}'$ does not cross itself, as it can be resolved as shown in \cref{fig:self-crossing}.
    Hence, they separate the plane into $4k|S|$ regions similarly to \cref{fig:reduction:distance-to-path-forest:regions}.
    For each $0 \leq i \leq 4k|S|$, let $R_i$ be the region that is bounded by two paths $P_{4ki}$ and $P_{4k(i+1)}$, where $P_{4k \cdot 4k|S|} \coloneqq P_0$.
    Observe that, for each vertex $v \in S \setminus \{s_1, s_2\}$, the number of paths in $\mathcal{P}'$ crossed by an edge incident to $N_{G'}(v)$ is at most $4k$:
    when $v$ belongs to region~$R_i$, each edge incident to a vertex in $N_{G'}(v)$ can only be in regions at most $2k$ apart from $R_i$, namely, $\{R_{j \bmod 4k|S|} \mid i-2k \leq j \leq i + 2k\}$.
    Since $S$ is a vertex cover of $G'$, for every edge $e \in E(G')$, $e = \{s_1, s_2\}$ (if it exists); $e$ is contained in a path in $\mathcal P$; or $e$ is incident to some vertex in $S\setminus\{s_1, s_2\}$ due to the assumption that there is no degree-1 vertex.
    Since at most $k$ paths in $\mathcal P'$ can be crossed by $e = \{s_1, s_2\}$ (if it exists) and at most $4k(|S|-2)$ paths in $\mathcal P'$ can be crossed by an edge incident to a vertex in $S \setminus \{s_1, s_2\}$, there are at least $4k|S| - (4k(|S| - 2)) - k = 7k$ paths in $\mathcal{P}'$ that are crossed only by the paths in $\mathcal{P}$ or not crossed at all.
    
    Let $P \in \mathcal{P}'$ be such a path.
    We remove all the paths in $\mathcal{P}$ but $P$ from $\Gamma$.
    Observe that now $P$ has no crossings at all.
    Hence, we can redraw all the removed paths along $P$ without making a crossing.
    We can also add the path $(s_1, u, s_2)$ to $\Gamma$ along $P$ in the same way and obtain a $k$-planar drawing of $G$.
\end{proof}

By summarizing the above lemmas, we can obtain a kernelization algorithm.
\begin{restatable}[\restateref{thm:vc-kernel}]{theorem}{ThmVcKernel}
\label{thm:vc-kernel}
    \kPTest{} has a kernel with $O(\vc(G)^2k^2\sqrt k)$ vertices.
    Moreover, such a kernel can be computed in linear time.
\end{restatable}

We can obtain a similar result for the case of neighborhood diversity.

\begin{restatable}[\restateref{lem:nd_k-planar}]{lemma}{LemNdkPlanar}
    \label{lem:nd_k-planar}
  The vertex cover number of a $k$-planar graph $G$ is at most $O(\nd(G) \sqrt k)$.
\end{restatable}

By plugging \cref{lem:nd_k-planar} to \cref{thm:vc-kernel}, we have the following.
\begin{corollary} \label{thm:nd-kernel}
    \kPTest{} has a kernel with $O(\nd(G)^2 k^3 \sqrt k)$ vertices.
\end{corollary}

Finally, we remark that the results of \cref{thm:vc-kernel} and \cref{thm:nd-kernel} cannot be generalized to vertex integrity and modular-width parameterizations (see \cite{GajarskyLO13,DBLP:journals/tcs/GimaHKKO22} for their definitions) unless $\mathrm{NP}\subseteq \mathrm{coNP} / \mathrm{poly}$, respectively.
Let $G_1, \ldots, G_t$ be graphs and $G$ be the disjoint union of $G_1,\ldots, G_t$.
Let $\vi(G)$ be the vertex integrity of $G$ and let $\mw(G)$ be the modular-width of $G$.
Then, it holds that (1) $G$ is 1-planar if and only if $G_i$ is 1-planar for all $1 \le i \le t$ and (2) $\vi(G) \le \max_{i}|V(G_i)|$ and $\mw(G) \le \max_i |V(G_i)|$.
This implies that \OPTest{} parameterized by vertex integrity and by modular-width is AND-cross-compositional~\cite{BodlaenderDFH09,Drucker15,Kernelization}, implying that unless $\mathrm{NP}\subseteq \mathrm{coNP} / \mathrm{poly}$, there is no polynomial kernelization for \OPTest{} of size $\mathrm{poly}(\vi(G))$ or $\mathrm{poly}(\mw(G))$.
As $\td(G) \le \vi(G)$, this kernelization lower bound also holds for treedepth parameterization.

%% file: appendix.tex
\section{Appendix: Missing Proofs}

\lemSubdivision*\label{lem:subdivision*}

\begin{proof}
    Let $\Gamma$ be a $k$-planar drawing of $G$.
    For each curve in $\Gamma$ representing an edge, since there are at most $k$ crossings on it, we can partition it into exactly $k$ subcurves so that each subcurve has at most one crossing on its interior.
    We obtain a $1$-planar drawing of $G_k$ by putting vertices on such $k-1$ division points for each curve.

    Conversely, let $\Gamma_k$ be a 1-planar drawing of $G_k$.
    For each $e = \{u, v\} \in E(G)$, let $P_e$ be the corresponding path $P_e$ between $u$ and $v$ in $G_k$, which consists of $k$ edges.
    We can assume that no two edges in $P_e$ cross, as otherwise we can remove the crossing without losing 1-planarity (\cref{fig:self-crossing}).
    Let $\Gamma$ be a drawing of $G$ obtained from $\Gamma_k$ by, for each $e \in E(G)$, concatenating the curves representing the edges in $P_e$.
    Since the vertices of $P_e$ have degree~$2$ except for the endpoints $u$ and $v$, concatenating the curves does not yield a new crossing.
    Hence, every curve representing an edge in $\Gamma$ has at most $k$ crossings.
\end{proof}

\LemCrossingFreeSpokes*\label{lem:crossing-free-spokes*}

\begin{proof}
    Let $\Gamma$ be a $k$-planar drawing of $G'$.
    Since each edge in $G$ has at most $k$ crossings and there are $km+1$ spokes between $r$ and each $u \in X$, there is at least one spoke connecting $r$ and $u$ that does not cross any edge in $G$.
    For $u \in X$, let $P_u$ be such a spoke between $r$ and $u$.
    Note that $P_u$ may cross other spoke edges.
    
    Suppose that $P_u$ does not have any crossings with $P_v$ in $\Gamma$ for all pairs $u, v \in X$.
    Then, we can redraw spokes between $r$ and $u$ other than $P_u$ along with the curve representing $P_u$ for each $u \in X$.
    The drawing obtained $\Gamma'$ in this way is still $k$-planar, and each spoke between $r$ and $u$, not limited to $P_u$, has no crossings at all.
    Thus, $\Gamma'$ is a $k$-planar drawing being claimed.
    
    Suppose otherwise.
    Let $u, v \in X$ such that $P_u$ crosses $P_v$ in $\Gamma$ and let $c$ be the crossing point.
    We can remove the crossing by exchanging the subcurve of $P_u$ between $r$ and $c$ with that of $P_v$ between $r$ and $c$ (\cref{fig:uncrossing-spoke}).
    \begin{figure}[t]
        \centering
        \includegraphics{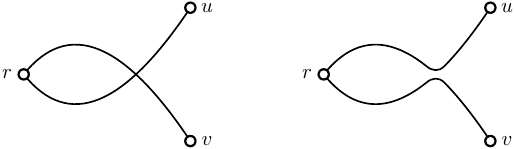}
        \caption{A crossing between $P_u$ and $P_v$ can be removed by rerouting subcurves separating the crossing point.}
        \label{fig:uncrossing-spoke}
    \end{figure}
    Note that the drawing obtained in this way may not be $k$-planar.
    By applying this uncrossing operation between $P_u$ and $P_v$ for all $u, v \in X$, we can obtain a (not necessarily $k$-planar) drawing of $G'$ such that each $P_u$ has no crossing with $P_v$ for all $v \in X$.
    Now, we can construct a drawing $\Gamma'$ of $G'$ by redrawing all spokes between $r$ and $u$ other than $P_u$ along with the curve representing $P_u$.
    The drawing $\Gamma'$ is indeed $k$-planar as all the spokes have no crossings and the subdrawings of $\Gamma$ and $\Gamma'$ induced by $G$ are identical.
    This completes the proof of the lemma.
\end{proof}

\LemUBPModerateSizeItem*\label{lem:UBP:moderate-size-item*}

\begin{proof}
    Starting from an instance of \UnaryBinPacking{} $I = \langle S, B, b \rangle$, we construct an equivalent instance $I' = \langle S', B', b \rangle$ such that every $x' \in S'$ satisfies $x' \leq \sqrt{B'}-1$.
    Without loss of generality, we assume that $x \le B$ for $x \in S$.
    We then add $2bB$ copies of $B+1$ to $S$, and let $B' = B + 2B(B+1)$.
    Let $S'$ be the multiset of integers obtained in this way.
    As $b$ and $B$ are bounded by a polynomial in $|S|$ from above, $I'$ is indeed an instance of \UnaryBinPacking{}.
    For every $x \in S'$, we can verify the restriction as:
    \begin{align*}
        x \le B = (B + 1) - 1 = \sqrt{B^2 + 2B + 1} - 1 \le \sqrt{2B^2 + 3B} - 1 = \sqrt{B'} - 1. 
    \end{align*}

    Now we show the correctness of the reduction.
    It is clear that if $I$ is a yes-instance then $I'$ is also a yes-instance.
    Conversely, let us assume that $I'$ is a yes-instance.
    Then, each bin must have exactly $2B$ copies of $B+1$, since there are $2bB$ of them and each bin can have at most $2B$ of them by construction.
    Hence, removing those copies from a solution for $I'$, all bins have integers whose total sum is exactly $B$, which forms a solution for $I$.
\end{proof}

\LemFptTdGk*\label{lem:fpt:td-Gk*}

\begin{proof}
    Let $F$ be an elimination forest of $G$ with height $d$.
    Now, we create an elimination forest of $G_k$ with height at most $d + \ceils{ \log_2 k}$.
    Let $(u, p_1, \ldots, p_{k-1}, v)$ be the path of $G_k$ corresponding to an edge $\{u,v\} \in E(G)$.
    Since the treedepth of the path graph $P_{n-1}$ with $n-1$ vertices is $\ceils{\log_2 n}$~\cite{SparsityNM},
    there is an elimination tree $T_{u,v}$ of $G_k[\{p_1, \ldots, p_{k-1}\}]$ with height $\ceils{\log_2 k}$.
    For each edge $\{u,v\} \in E(G)$, assume that $u$ is an ancestor of $v$ and attach $T_{u,v}$ as a child subtree to node $v$ in $F$. Let $F_k$ be the rooted forest obtained as above.
    It is clear that the height of $F_k$ is at most $d + \ceils{\log_2 k}$.
    Moreover, for each edge of $E(G_k)$, one of the endpoints is an ancestor of the other end in $F_k$.
    Thus, $F_k$ is an elimination forest of $G_k$ with height at most $d + \ceils{\log_2 k}$.
\end{proof}

\ThmVcKernel*\label{thm:vc-kernel*}
\begin{proof}
Since removing vertices of degree at most~1 does not change the local crossing number, we can assume that the minimum degree of $G$ is at least~2.
Let $S$ be a vertex cover of $G$.
We define a labeling $\pi: V(G) \setminus S \to \binom{N_G(v)}{2}$ that maps $v \in V(G) \setminus S$ to an unordered pair of its neighbors, and compute $\#_{\pi}(V(G) \setminus S)$ in linear time by radix sort.
From \cref{lem:vc:deg2_class}, if $\#_{\pi}(V(G) \setminus S) > 3.81|S|\sqrt{2k}$ then $G$ is not $k$-planar.

Suppose that there are more than $7k(|S| - 2)$ vertices of degree at least~3 in $V(G) \setminus S$ having the same labeling on $\pi$.
Since $|S\setminus \pi(v)| = |S|-2$, by the pigeonhole principle, $G$ contains $K_{7k+1, 3}$ as a subgraph.
Hence, by \cref{cor:vc:false_twins-deg3}, we can assume that $G$ has at most $7k(|S| - 2)\cdot 3.81|S|\sqrt{2k} \in O(|S|^2 k\sqrt k)$ vertices of degree at least 3 in $V(G)\setminus S$.
Applying \cref{lem:vc:false_twins-deg2} exhaustively, the number of degree-2 vertices is at most $16k^2|S| \cdot 3.81|S| \sqrt{2k} \in O(|S|^2 k^2\sqrt k)$.
Thus, we conclude that $|V(G)| \in O(|S|^2 k^2 \sqrt k)$.

Moreover, each processing can be done by counting vertices on the labeling based on $\pi$ and the degree, and removing some of them. Therefore, the entire running time of our kernelization is linear in the size of $G$.
\end{proof}

\LemNdkPlanar*\label{lem:nd_k-planar*}
To show the \cref{lem:nd_k-planar}, we first show the following lemma.

\begin{lemma}
\label{lem:nd-sparse-vc-bound}
    Let $G$ be a graph with neighborhood diversity $d$.
    If $G$ does not contain $K_{t,t}$ as a subgraph for some~$t$, then the vertex cover number of $G$ is at most $2dt$.
\end{lemma}

\begin{proof}
    We construct a vertex cover $S$ with at most $2dt$ vertices of $G$ in the following manner. 
    Suppose that the vertex set of $G$ is partitioned into twin classes $\{V_1, \dots, V_d\}$.
    For each true twin class $V_i$, we put all the vertices in $V_i$ into $S$.
    Since $V_i$ induces a clique in $G$ and $G$ does not have $K_{t, t}$ as a subgraph, we have $|V_i| < 2t$.
    If $S$ is a vertex cover of $G$, we are done.
    Suppose that there is an edge that is not covered by $S$.
    As each true twin class is already contained in $S$, this edge spans two distinct false twin classes $V_i$ and $V_j$.
    Since $G$ has no $K_{t, t}$, at least one of $V_i$ and $V_j$, say $V_i$, has less than $t$ vertices.
    We then add all the vertices in $V_i$ to $S$.
    As long as there is an edge not covered by $S$, we continue this process, and then we eventually have a vertex cover $S$ of $G$ with $|S| \leq 2dt$.
\end{proof}

\begin{proof}[Proof of \cref{lem:nd_k-planar}]
    Let $t = 8\ceils[\big]{\sqrt{2k}}$.
    By~\cref{lem:nd-sparse-vc-bound}, it suffices to show that $K_{t, t}$ is not $k$-planar.
    Clearly, $K_{t, t}$ has $16 \ceils[\big]{\sqrt{2k}}$ vertices and $4\ceils[\big]{\sqrt{2k}} \cdot 16 \ceils[\big]{\sqrt{2k}}$ edges.
    Since every $k$-planar graph with $n$ vertices has at most $3.81n\sqrt{2k}$ edges~\cite{Ackerman19}, $K_{t,t}$ is not $k$-planar.
\end{proof}